%% file: reconfiguration.tex
\documentclass[a4paper,USenglish]{lipics-v2018}
\usepackage{microtype}%if unwanted, comment out or use option "draft"

\usepackage[table,xcdraw]{xcolor}

\usepackage{tikz}

\usepackage{todonotes}

\usepackage[table,xcdraw]{xcolor}

\usepackage{color}
\definecolor{lightblue}{rgb}{0.5,0.5,1.0}
\definecolor{darkred}{rgb}{0.8,0,0}
\definecolor{darkgreen}{rgb}{0,0.5,0}
\definecolor{darkblue}{rgb}{0,0,0.5}

\usepackage{algorithm}
\usepackage[noend]{algpseudocode}
\title{Token Sliding on Split Graphs\footnote{Supported by JSPS and MAEDI
under the Japan-France Integrated Action Program (SAKURA) Project GRAPA
38593YJ, by FMJH program PGMO and EDF via project
2016-1760H/C16/1507 ``Stability versus Optimality in Dynamic Environment
Algorithmics'' and project ``ESIGMA'' (ANR-17-CE23-0010),
and by JSPS KAKENHI Grant Numbers JP18K11157, JP18K11168, JP18K11169, JP18H04091.}}

\author{Rémy Belmonte}{University of Electro-Communications, Chofu, Tokyo, 182-8585, Japan}{remybelmonte@gmail.com}{0000-0001-8043-5343}{}%mandatory, please use full name; only 1 author per \author macro; first two parameters are mandatory, other parameters can be empty.

\author{Eun Jung Kim}{Université Paris-Dauphine, PSL University, CNRS, LAMSADE, 75016, Paris, France}{eun-jung.kim@dauphine.fr}{}{}%mandatory, please use full name; only 1 author per \author macro; first two parameters are mandatory, other parameters can be empty.

\author{Michael Lampis}{Université Paris-Dauphine, PSL University, CNRS, LAMSADE, 75016, Paris, France}{michail.lampis@lamsade.dauphine.fr}{0000-0002-5791-0887}{}%mandatory, please use full name; only 1 author per \author macro; first two parameters are mandatory, other parameters can be empty.

\author{Valia Mitsou}{Université Paris-Diderot, IRIF, CNRS, 75205, Paris, France}{vmitsou@liris.cnrs.fr}{}{}%mandatory, please use full name; only 1 author per \author macro; first two parameters are mandatory, other parameters can be empty.

\author{Yota Otachi}{Kumamoto University, Kumamoto, 860-8555, Japan}{otachi@cs.kumamoto-u.ac.jp}{0000-0002-0087-853X}{}%mandatory, please use full name; only 1 author per \author macro; first two parameters are mandatory, other parameters can be empty.

\author{Florian Sikora}{Université Paris-Dauphine, PSL University, CNRS, LAMSADE, 75016, Paris, France}{florian.sikora@dauphine.fr}{}{}%mandatory, please use full name; only 1 author per \author macro; first two parameters are mandatory, other parameters can be empty.

\authorrunning{R. Belmonte, E.J. Kim, M. Lampis, V. Mitsou, Y. Otachi, F. Sikora}%mandatory. First: Use abbreviated first/middle names. Second (only in severe cases): Use first author plus 'et. al.'

\Copyright{Rémy Belmonte, Eun Jung Kim, Michael Lampis, Valia Mitsou, Yota Otachi, Florian Sikora}%mandatory, please use full first names. LIPIcs license is "CC-BY";  http://creativecommons.org/licenses/by/3.0/

\subjclass{; }% mandatory: Please choose ACM 2012 classifications from https://www.acm.org/publications/class-2012 or https://dl.acm.org/ccs/ccs_flat.cfm . E.g., cite as "General and reference $\rightarrow$ General literature" or \ccsdesc[100]{General and reference~General literature}. 

%mandatory

\category{}%optional, e.g. invited paper

\relatedversion{}%optional, e.g. full version hosted on arXiv, HAL, or other respository/website

\supplement{}%optional, e.g. related research data, source code, ... hosted on a repository like zenodo, figshare, GitHub, ...

\funding{}%optional, to capture a funding statement, which applies to all authors. Please enter author specific funding statements as fifth argument of the \author macro.

\acknowledgements{}%optional

\nolinenumbers %uncomment to disable line numbering
\hideLIPIcs  %uncomment to remove references to LIPIcs series (logo, DOI, ...), e.g. when preparing a pre-final version to be uploaded to arXiv or another public repository
\begin{document}

\maketitle

\begin{abstract} We consider the complexity of the \textsc{Independent Set
Reconfiguration} problem under the Token Sliding rule. In this problem we are
given two independent sets of a graph and are asked if we can transform one to
the other by repeatedly exchanging a vertex that is currently in the set with
one of its neighbors, while maintaining the set independent. Our main result is
to show that this problem is PSPACE-complete on split graphs (and hence also on
chordal graphs), thus resolving an open problem in this area.

We then go on to consider the $c$-\textsc{Colorable Reconfiguration} problem
under the same rule, where the constraint is now to maintain the set
$c$-colorable at all times. As one may expect, a simple modification of our
reduction shows that this more general problem is PSPACE-complete for all fixed
$c\ge 1$ on chordal graphs. Somewhat surprisingly, we show that the same cannot
be said for split graphs: we give a polynomial time ($n^{O(c)}$) algorithm for
all fixed values of $c$,  except $c=1$, for which the problem is
PSPACE-complete.  We complement our algorithm with a lower bound showing that
$c$-\textsc{Colorable Reconfiguration} is W[2]-hard on split graphs
parameterized by $c$ and the length of the solution, as well as a tight
ETH-based lower bound for both parameters. 

%	Finally, we show
%that the same problem is in P for strongly chordal graphs.

\end{abstract}

\input{intro.tex}

\input{defs.tex}

\input{pspace-main.tex}

\input{chordal.tex}

\input{xpalgo.tex}

\input{whard.tex}

%\input{scg.tex}

%%
%% Bibliography
%%

%% Please use bibtex, 

\bibliography{reconfiguration}

%\newpage
%
%\appendix
%
%\section{Omitted Material}
%
%\input{appendix.tex}

\end{document}

%% file: intro.tex
\section{Introduction}

A reconfiguration problem is a problem of the following type: we are given an
instance of a decision problem, two feasible solutions $S,T$, and a local
modification rule. The question is whether $S$ can be transformed to $T$ by
repeated applications of the modification rule in a way that maintains the
solution feasible at all times. Due to their numerous applications,
reconfiguration problems have attracted much interest in the literature, and
reconfiguration versions of standard problems (such as \textsc{Satisfiability},
\textsc{Dominating Set}, and \textsc{Independent Set}) have been widely studied
(see the surveys \cite{vandenHeuvel13,Nishimura18} and the references therein).

Among reconfiguration problems on graphs, \textsc{Independent Set
Reconfiguration} is certainly the most well-studied. The complexity of this
problem depends heavily on the rule specifying the allowed reconfiguration
moves.  The main reconfiguration rules that have been studied for
\textsc{Independent Set Reconfiguration} are Token Addition \& Removal (TAR)
\cite{KaminskiMM12,MouawadN0SS17}, Token Jumping (TJ)
\cite{BonsmaKW14,BousquetMP17,ItoDHPSUU11,ItoKOSUY14,ItoKO14}, and Token
Sliding (TS)
\cite{BonamyB17,DemaineDFHIOOUY15,Fox-EpsteinHOU15,HearnD05,HoangU16,LokshtanovM18}.
In all rules, we are required to keep the current set independent at all times.
TAR allows us to add or remove any vertex in the current set, as long as the
set's size is always higher than a predetermined threshold.  TJ allows to
exchange any vertex in the set with any vertex outside it (thus keeping the
size of the set constant at all times).  Finally, under TS, we are allowed to
exchange a vertex in the current independent set with one of its neighbors,
that is, we are allowed to perform a TJ move only if the two involved vertices
are adjacent.    

The \textsc{Independent Set Reconfiguration} problem has been intensively
studied under all three rules.  Because the problem is PSPACE-complete in
general for all three rules \cite{KaminskiMM12}, this has motivated the study
of its complexity in restricted classes of graphs, with an emphasis on graphs
where \textsc{Independent Set} is polynomial-time solvable, such as chordal
graphs and bipartite graphs.  By now, many results of this type have been
discovered (see Table~\ref{tbl:summary} for a summary). 

Our first, and main, focus of this paper is to concentrate on a case of this
problem which has so far remained elusive, namely, the complexity of
\textsc{Independent Set Reconfiguration} on chordal graphs under the TS rule.
This case is of particular interest because it is one of the few cases where
the problem is known to be tractable under both TAR and TJ. Indeed,
Kami\'{n}ski, Medvedev, and Milani\v{c}~\cite{KaminskiMM12} showed that under
these two rules \textsc{Independent Set Reconfiguration} is polynomial-time
solvable on even-hole-free graphs, a class that contains chordal graphs. In the
same paper they explicitly asked as an open question if the same problem is
tractable on even-hole-free graphs under TS (\cite[Question 2]{KaminskiMM12}).

This question was then taken up by Bonamy and Bousquet~\cite{BonamyB17} who
made some progress by showing that \textsc{Independent Set Reconfiguration}
under TS is polynomial-time solvable on \emph{interval graphs}, an important
subclass of chordal graphs. They also gave some first evidence that it may be
hard to obtain a similarly positive result for chordal graphs by showing that a
related problem, the problem of determining if \emph{all} independent sets of
the same size can be transformed to each other under TS, is coNP-hard on split
graphs, another subclass of chordal graphs.  Note, however, that this is a
problem that is clearly distinct from the more common reconfiguration problem
(which asks if two \emph{specific} sets are reachable from each other), and
that the coNP-hardness is not tight, since the best known upper bound for this
problem is also PSPACE.

The complexity of \textsc{Independent Set Reconfiguration} under TS on split
and chordal graphs has thus remained as an open problem. Our first, and main,
contribution in this paper is to settle this problem by showing that the
problem is PSPACE-complete already on split graphs (Theorem
\ref{thm:split-TS}), and therefore also on chordal and even-hole-free graphs.

\begin{table}[tbp]
\centering
\caption{Complexity of \textsc{Independent Set Reconfiguration} on some graph classes.}
\label{tbl:summary}
\begin{tabular}{c|c|c}
& \multicolumn{2}{c}{\cellcolor{lightgray!50}\textsc{Independent Set Reconfiguration}} \\
& \cellcolor{lightgray!25}TS & \cellcolor{lightgray!25}TJ/TAR \\ \hline
\cellcolor{lightgray!50}perfect & \multicolumn{2}{c}{PSPACE-complete \cite{KaminskiMM12}} \\ \hline
\cellcolor{lightgray!50}even-hole-free & PSPACE-complete (Theorem~\ref{thm:split-TS}) & P \cite{KaminskiMM12}  \\ \hline
\cellcolor{lightgray!50}chordal & PSPACE-complete (Theorem~\ref{thm:split-TS}) & P (even-hole-free) \\ \hline
\cellcolor{lightgray!50}split & PSPACE-complete (Theorem~\ref{thm:split-TS}) & P (even-hole-free) \\ \hline
\cellcolor{lightgray!50}interval & P \cite{BonamyB17} & P (even-hole-free) \\ \hline
\cellcolor{lightgray!50}bipartite & PSPACE-complete \cite{LokshtanovM18} & NP-complete \cite{LokshtanovM18} \\ \hline
\end{tabular}
\end{table}

\subparagraph{$c$-\textsc{Colorable Reconfiguration}} A natural generalization
of \textsc{Independent Set Reconfiguration} was recently introduced in
\cite{ItoO18}: in $c$-\textsc{Colorable Reconfiguration} we are given a graph
$G=(V,E)$ and two sets $S,T\subseteq V$, both of which induce a $c$-colorable
graph. The question is whether $S$ can be transformed to $T$ (under any of the
previously mentioned rules) in a way that maintains a $c$-colorable graph at
all times. Clearly, $c=1$ is the case of \textsc{Independent Set
Reconfiguration}. It was shown in \cite{ItoO18} that this problem is already
PSPACE-complete on split graphs under all three rules, when $c$ is part of the
input. It was thus posed as an open question what is the complexity of the same
problem when $c$ is fixed. Some first results in this direction were given in
the form of an $n^{O(c)}$ (XP) algorithm that works for split graphs under the
TAR and TJ rules (but not TS). Motivated by this work, the second area of focus
of this paper is to investigate how the hardness of $1$-\textsc{Colorable
Reconfiguration} for split graphs established in Theorem \ref{thm:split-TS}
extends to larger, but fixed $c$.

Our first contribution in this direction is to show that, for chordal graphs,
$c$-\textsc{Colorable Reconfiguration} under TS is PSPACE-complete for any
fixed $c\ge 1$. This is, of course, not surprising, as the problem is
PSPACE-complete for $c=1$; indeed, the reduction we present in Theorem
\ref{thm:TS-split2} is a tweak of the construction of Theorem
\ref{thm:split-TS} that increases $c$.

What is perhaps more surprising is that we show (under standard assumptions)
that, even though Theorem \ref{thm:split-TS} establishes hardness for $c=1$ on
split graphs, a similar tweak cannot establish hardness for higher $c$ on the
same class for TS. Indeed, we provide an algorithm which solves TS
$c$-\textsc{Colorable Reconfiguration} in split graphs in time $n^{O(c)}$ for
any $c$ \emph{except} $c=1$. Thus, \textsc{Independent Set Reconfiguration}
turns out to be the \emph{only} hard case of $c$-\textsc{Colorable
Reconfiguration} for split graphs under TS.  Since the $n^{O(c)}$ algorithm of
\cite{ItoO18} for TAR/TJ reconfiguration of split graphs works for all fixed
$c$, it  thus seems that this anomalous behavior is peculiar to the Token
Sliding rule.

Finally, we address the natural question of whether one can improve this
$n^{O(c)}$ algorithm, by showing that the problem is W[2]-hard parameterized by
$c$ and the length of the solution $\ell$ for all three rules.  This is in a
sense doubly tight, since in addition to our algorithm and the algorithm of
\cite{ItoO18} which run in $n^{O(c)}$, it also matches the trivial
$n^{O(\ell)}$ algorithm which tries out all solutions of length $\ell$. More
strongly, under the ETH our reduction implies that the problem cannot be solved
in $n^{o(c+\ell)}$ meaning that these algorithms are in a sense ``optimal''.

%Finally, to complement these mostly negative results, we consider a sub-class
%of chordal graphs that is incomparable to split graphs, namely strongly chordal
%graphs.  For this class of graphs we show that the $c$-\textsc{Colorable
%Reconfiguration} problem is in P even if we are looking for a shortest
%reconfiguration sequence, by showing how one can always construct a solution
%which at every step decreases the symmetric difference of the current set from
%the target.

%% file: defs.tex
\section{Definitions}

%Graphs classes
We use standard graph-theoretic terminology.  For a graph $G=(V,E)$ and a set
$S\subseteq V$ we use $G[S]$ to denote the graph induced by $S$. A graph is
chordal if it does not contain a $k$-vertex cycle $C_k$ as an induced subgraph for any $k>3$. A
graph is split if its vertex set can be partitioned into two sets $K,I$ such
that $K$ induces a clique and $I$ induces an independent set. It is a
well-known fact that split graphs are chordal, and it is easy to see that both
classes are closed under induced subgraphs. We use $\chi(G),\omega(G)$ to
denote the chromatic number and maximum clique size of a graph $G$
respectively. It is known that, because chordal graphs are perfect, if $G$ is
chordal then $\chi(G)=\omega(G)$ \cite{west2001introduction}. We also recall
that a graph $G$ is chordal if and only if every induced subgraph of $G$
contains a simplicial vertex, where a vertex is simplicial if its neighborhood
is a clique.

%Problem defs

Let $G=(V,E)$ be a graph and $c\ge 1$ an integer. Given two sets $S,T\subseteq
V$ such that $\chi(G[S]),\chi(G[T])\le c$, we say that $S$ can be
$c$-transformed into $T$ by one token sliding (TS) move if $|T|=|S|$ and there
exist $u,v\in V$ with $(u,v)\in E$ such that $\{u\} =  T\setminus S$, $\{v\} =
S\setminus T$. One easy way to think of TS moves is by picturing the elements
of the current set $S$ as tokens placed on the vertices of the graph, and a
single move as ``sliding'' a token along an edge (hence the name Token
Sliding).

We say that $S$ is $c$-reachable from $T$, or that $S$ can be $c$-transformed
into $T$, by a sequence of TS moves if there exists a sequence of sets
$I_0,I_1,\ldots,I_{\ell}$, with $I_0=S, I_{\ell}=T$ and for each
$i\in\{0,\ldots,\ell-1\}$, $\chi(G[I_i])\le c$ and $I_{i}$ can be
$c$-transformed into $I_{i+1}$ by one TS move. We will simply say that $S$ can
be transformed into $T$ or that $S$ is reachable from $T$, if $S,T$ are
independent sets and $S$ can be $1$-transformed into $T$.  We focus on the
following problems.

\begin{definition}

In $c$-\textsc{Colorable Reconfiguration} we are given  a graph $G=(V,E)$ and
two sets $S,T\subseteq V$ with $|S|=|T|$ and $\chi(G[S]), \chi(G[T])\le c$.  We
are asked if $S$ can be $c$-transformed into $T$. \textsc{Independent Set
Reconfiguration} is the special case of $c$-\textsc{Colorable Reconfiguration}
where $c=1$.

\end{definition}

%We remark that clearly $1$-\textsc{Colorable Reconfiguration} is the same
%problem as \textsc{Independent Set Reconfiguration}.

In addition to TS moves we will consider Token Jumping (TJ) and Token Addition
\& Removal (TAR) moves. A TJ move is the same as a TS move except that the two
vertices $u,v$ are not required to be adjacent. Two $c$-colorable sets $S,T$
are reachable with one TAR move with threshold $k$ if $|S|,|T|\ge k$ and
$|(S\setminus T)\cup(T\setminus S)|=1$. We note here that, because our main
focus in this paper is the TS rule, whenever we refer to a transformation
without explicitly specifying under which rule this transformation is performed
the reader may assume that we are referring to the TS rule.

%Complexity and FPT

We assume that the reader is familiar with basic complexity notions such as the
class PSPACE \cite{papad}, as well as basic notions in parameterized
complexity, such as the class W[2] (see e.g.~\cite{CyganFKLMPPS15}).  In
Theorem \ref{thm:split-TS} we will perform a reduction from the PSPACE-complete
NCL (non-deterministic constraint logic) reconfiguration problem introduced by
Demaine and Hearn in \cite{HearnD05} (see also \cite{Hearn06,HD09}). Let us
recall this problem. In the NCL reconfiguration problem we are given as input a
graph $G=(V,E)$, whose edge set is partitioned into two sets, $R$ (red) and $B$
(blue). We consider blue edges as edges of weight $2$ and red edges as edges of
weight $1$. A valid configuration of $G$ is an orientation of all the edges
with the property that all vertices have weighted in-degree at least $2$.  In
the NCL configuration-to-configuration problem we are given two valid
orientations of $G$, $D$ and $D'$, and are asked if there is a sequence of
valid orientations $D_0, D_1, \ldots, D_t$ such that $D=D_0, D'=D_t$ and for
all $i\in\{0,\ldots,t-1\}$ we have that $D_i, D_{i+1}$ agree on all edges
except one. We recall the following theorem:

\begin{theorem}[Corollary 6 of \cite{HearnD05}]\label{thm:ncl} The NCL
configuration-to-configuration problem is PSPACE-complete even if all vertices
of $G$ have degree exactly three and, moreover, even if all vertices belong in
one of the following two types: OR vertices, which are vertices incident on
exactly three blue edges and no red edges; and AND vertices which are vertices
incident on two red edges and one blue edge.  \end{theorem}

%% file: pspace-main.tex
\section{Token Sliding on Split Graphs is PSPACE-complete}\label{sec:PSPACE-split}

The main result of this section is that \textsc{Independent Set
Reconfiguration} is PSPACE-complete under the TS rule when restricted to split
graphs. 

\subsubsection*{Overview of the proof} Our proof is a reduction from the NCL
(non-deterministic constraint logic) reconfiguration problem of Theorem
\ref{thm:ncl}.  The first step of our proof is a relatively straightforward
reduction from the NCL reconfiguration problem to token sliding on split
graphs. Its main idea is roughly as follows: for each edge $e=(u,v)$ of the
original graph we construct two selection vertices $e_u, e_v$ in the
independent set of our split graph.  The idea is that at each point exactly one
of the two will contain a token (i.e. will belong in the current independent
set), hence our independent set will in a natural way represent an orientation
of the original graph. In order to allow a single reconfiguration step to take
place we add for each pair of selection vertices $e_u,e_v$ one or two ``gate''
vertices (depending on the color of $e$), which are common neighbors of
$e_u,e_v$ and belong in the clique. The idea is that a single re-orientation
step would, for example, take a token from $e_u$, slide it to a gate vertex
connected to the pair $e_u,e_v$, and then slide it to $e_v$: this sequence
would represent re-orienting $e$ from $u$ to $v$.  In order to simulate the
in-degree constraint we add edges between each selection vertex $e_u$ and gate
vertices corresponding to edges incident on \emph{the other} endpoint of $e$,
since keeping a token on $e_u$ represents an orientation of $e$ towards $u$,
which makes it harder to re-orient the edges incident on the other endpoint of
$e$.

The above sketch captures the basic idea of our reduction, except for one
significant obstacle. The correspondence between orientations and independent
sets is only valid if we can guarantee that no intermediate independent set
will ``cheat'' by, for example, placing tokens on both $e_u$ and $e_v$.  Since
we have added edges from $e_u,e_v$ to gate vertices that correspond to other
edges (in order to simulate the interaction between edges in the NCL instance),
nothing prevents a reconfiguration solution from using these edges to slide a
token from one selection pair to another.  The main problem thus becomes
enforcing consistency, or in other words forcing the solution sequence to only
use the appropriate gate vertices to slide tokens as intended.  This is handled
in the second step of our reduction which, given the split graph construction
sketched above, makes a large number of copies and connects them appropriately
in a way that the only feasible token sliding solutions are indeed those that
correspond to valid orientations of the original graph.

In the remainder of this section we use the following notation: $G=(V,E)$,
where $E=R\cup B$, is the graph supplied with the initial NCL reconfiguration
instance and $D,D'$ are the initial and target orientations; $G_b=(V_b,E_b)$ is
the ``basic'' split graph of our construction in the first step and $S,T$ the
independent sets of $G_b$ for which we need to decide reachability; and
$G_f=(V_f,E_f)$ is the split graph of our final token sliding instance with
$S_f,T_f$ being its corresponding independent sets.

Before we proceed, let us first slightly edit our given NCL reconfiguration
instance. We will now allow some vertices to have degree two and call these
vertices COPY vertices. Using these we can force the OR vertices to become an
independent set.

\begin{lemma}\label{lem:mycsl}

NCL reconfiguration remains PSPACE-complete on graphs where (i) all vertices
are either AND vertices (two incident red edges, one incident blue edge), OR
vertices (three incident blue edges), or COPY vertices (two incident blue
edges) (ii) every blue edge is incident on exactly one COPY vertex.

\end{lemma}

\begin{proof}

For every blue edge $e=(u,v)\in B$ in the original graph we delete this edge
from the graph, introduce a new COPY vertex $w$, and connect $w$ to $u,v$ with
blue edges. It is not hard to see that this transformation does not change the
type of any original vertex or the answer to the reconfiguration problem.
\end{proof}

\subsubsection*{First Step of the Construction} 

We assume (Lemma \ref{lem:mycsl}) that in the given graph $G$ we have three
types of vertices (AND, OR, COPY) and that each blue edge is incident on one
COPY vertex.  Let us now describe the construction of $G_b$.

\begin{enumerate}

\item For each $e=(u,v)\in R$ we construct two selector vertices $e_u,e_v$ and
one gate vertex $g_e$.

\item For each $e=(u,v)\in B$ we construct two selector vertices $e_u,e_v$ and
two gate vertices $g_{e,1}, g_{e,2}$.

\item For each edge $e=(u,v)\in R$ we connect $g_e$ to both $e_u,e_v$. For each
edge $e=(u,v)\in B$ we connect both $g_{e,1},g_{e,2}$ to both $e_u,e_v$. We
call the edges added in this step gate edges.

\item For each AND vertex $u$, such that $e=(u,v_1)\in B$ and $f=(u,v_2)\in R,\
h= (u,v_3)\in R$ we add the following edges: $(e_{v_1},g_f), (e_{v_1},g_h),
(f_{v_2}, g_{e,1}), (f_{v_2}, g_{e,2}), (h_{v_3}, g_{e,1}), (h_{v_3}, g_{e,2})$
(see Figure \ref{fig:ANDOR}). In other words, for each edge involved in this part we
connect the selector which represents its other endpoint (not $u$) to the gate
vertices of edges that should be unmovable if this edge is not oriented towards
$u$.

\item For each OR vertex $u$ such that $e=(u,v_1), f=(u,v_2), h=(u,v_3)\in B$
we add the following edges: $(e_{v_1}, g_{f,1}), (e_{v_1}, g_{h,1}), (e_{v_2},
g_{e,1}), (e_{v_2}, g_{h,2}), (e_{v_3}, g_{e,2}), (e_{v_3}, g_{f,2})$. In other
words, we connect the selector vertex for each $v_i$ to a distinct gate of the
edges $(u,v_j), (u,v_k)$, for $i,j,k$ distinct. Informally, this makes sure
that if two of the edges are oriented away from $u$ the third edge is stuck,
but if at most one is oriented away from $u$ the other edges have a free gate.

\item For each COPY vertex $u$ such that $e=(u,v_1), f=(u,v_2)\in B$ we add the
following edges: $(e_{v_1}, g_{f,1}), (e_{v_1}, g_{f,2}), (f_{v_2}, g_{e,1}),
(f_{v_2}, g_{e_2})$. In other words, we connect the selector vertex for $v_1$
in a way that blocks the movement of the token from $f_u$, and similarly for
$v_2$.

\item We connect all gate vertices into a clique to obtain a split graph. Note
that the remaining vertices (that is, the selector vertices $e_v$) form an
independent set.

\end{enumerate}

We now construct two independent sets $S,T$ of $G_b$ in the natural way: given
an orientation $D$, for each $e=(u,v)$ we place $e_u$ in $S$ if and only if $D$
orients $e$ towards $u$; we construct $T$ from $D'$ in the same way. This
completes the basic construction.

Before proceeding, let us make some basic observations regarding the
neighborhoods of gate vertices of the graph $G_b$. We have the following:

\begin{itemize}

\item If $e=(u,v)\in R$, let $u',v'$ be vertices of $G$ such that $f=(u,u')\in
B$, $h=(v,v')\in B$ (that is, $u',v'$ are the second endpoints of the blue
edges incident on $u,v$). We have that $N(g_e) = \{ e_u,e_v, f_{u'}, h_{v'}\}$.

\item If $e=(u,v)\in B$, $u$ is a COPY vertex and $v$ is an AND vertex, let
$f=(u,u')\in B$ be the other edge incident on $u$, and $h=(v,v'),
\ell=(v,v'')\in R$ be the other two edges incident on $v$. Then $N(g_{e,1})=
N(g_{e,2}) = \{ e_u, e_v, f_{u'}, h_{v'}, \ell_{v''}\} $.

\item If $e=(u,v)\in B$, $u$ is a COPY vertex and $v$ is an OR vertex, let
$f=(u,u')\in B$ be the other edge incident on $u$, and $h=(v,v'),
\ell=(v,v'')\in B$ be the other two edges incident on $v$. Then one of the
vertices $g_{e,1}, g_{e,2}$ has neighbors $\{e_u, e_v, f_{u'}, h_{v'}\}$ and
the other has neighbors $\{e_u, e_v, f_{u'}, \ell_{v''}\}$.

\end{itemize}

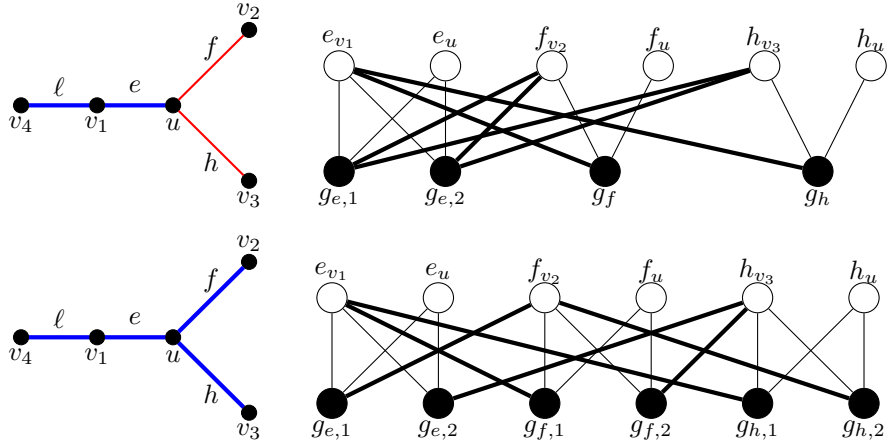
\begin{figure}[ht]
\begin{center}

\begin{tabular}{lr}

\begin{tikzpicture}[scale=0.5]
\tikzstyle{sommet}=[circle, draw, fill=black, inner sep=1pt, inner sep=4pt,
minimum size=0.1cm, scale=0.5]
\tikzstyle{redline}=[red, thick]
\tikzstyle{blueline}=[blue, ultra thick]

\node[sommet] (u) at (0,0) {}; 
\node[sommet] (v) at (2,0) {};
\node[sommet] (up) at (-2,0) {};
\node[sommet] (v1) at (4,2) {};
\node[sommet] (v2) at (4,-2) {};

\draw[redline] (v) -- (v1);
\draw[redline] (v) -- (v2);
\draw[blueline] (u) -- (v);
\draw[blueline] (u) -- (up);

\node () at (1,0.5) {$e$};
\node () at (3,1.5) {$f$};
\node () at (3,-1.5) {$h$};
\node () at (-1,0.5) {$\ell$};
\node () at (2,-0.5) {$u$};
\node () at (0,-0.5) {$v_1$};
\node () at (4,2.5) {$v_2$};
\node () at (4,-2.5) {$v_3$};
\node () at (-2,-0.5) {$v_4$};

\end{tikzpicture}

&

\begin{tikzpicture}[scale=0.7]
\tikzstyle{clique}=[circle, draw, fill=black, inner sep=1pt, inner sep=4pt, minimum size=0.1cm]
\tikzstyle{ind}=[circle, draw, inner sep=1pt, inner sep=4pt, minimum size=0.1cm]
\tikzstyle{conflict}=[ultra thick]

\node[clique] (gf) at (0,0) {};
\node[ind] (fu) at (1,2) {};
\node[ind] (fv2) at (-1,2) {};
\node () at (0,-0.5) {$g_f$};
\node () at (1,2.5) {$f_u$};
\node () at (-1,2.5) {$f_{v_2}$};

\draw (gf) -- (fu);
\draw (gf) -- (fv2);

\node[clique] (gh) at (4,0) {};
\node[ind] (hu) at (5,2) {};
\node[ind] (hv3) at (3,2) {};
\node () at (4,-0.5) {$g_h$};
\node () at (5,2.5) {$h_u$};
\node () at (3,2.5) {$h_{v_3}$};

\draw (gh) -- (hu);
\draw (gh) -- (hv3);

\node[clique] (ge) at (-5,0) {};
\node[clique] (ge2) at (-3,0) {};
\node[ind] (eu) at (-3,2) {};
\node[ind] (ev1) at (-5,2) {};
\node () at (-5,-0.5) {$g_{e,1}$};
\node () at (-3,-0.5) {$g_{e,2}$};
\node () at (-3,2.5) {$e_u$};
\node () at (-5,2.5) {$e_{v_1}$};

\draw (ge) -- (eu);
\draw (ge) -- (ev1);
\draw (ge2) -- (eu);
\draw (ge2) -- (ev1);

\draw[conflict] (ev1) -- (gf);
\draw[conflict] (ev1) -- (gh);

\draw[conflict] (fv2) -- (ge);
\draw[conflict] (fv2) -- (ge2);

\draw[conflict] (hv3) -- (ge);
\draw[conflict] (hv3) -- (ge2);

\end{tikzpicture}

\\

\begin{tikzpicture}[scale=0.5]
\tikzstyle{sommet}=[circle, draw, fill=black, inner sep=1pt, inner sep=4pt,
minimum size=0.1cm, scale=0.5]
\tikzstyle{redline}=[red, thick]
\tikzstyle{blueline}=[blue, ultra thick]

\node[sommet] (u) at (0,0) {}; 
\node[sommet] (v) at (2,0) {};
\node[sommet] (up) at (-2,0) {};
\node[sommet] (v1) at (4,2) {};
\node[sommet] (v2) at (4,-2) {};

\draw[blueline] (v) -- (v1);
\draw[blueline] (v) -- (v2);
\draw[blueline] (u) -- (v);
\draw[blueline] (u) -- (up);

\node () at (1,0.5) {$e$};
\node () at (3,1.5) {$f$};
\node () at (3,-1.5) {$h$};
\node () at (-1,0.5) {$\ell$};
\node () at (2,-0.5) {$u$};
\node () at (0,-0.5) {$v_1$};
\node () at (4,2.5) {$v_2$};
\node () at (4,-2.5) {$v_3$};
\node () at (-2,-0.5) {$v_4$};

\end{tikzpicture}

&

\begin{tikzpicture}[scale=0.7]
\tikzstyle{clique}=[circle, draw, fill=black, inner sep=1pt, inner sep=4pt, minimum size=0.1cm]
\tikzstyle{ind}=[circle, draw, inner sep=1pt, inner sep=4pt, minimum size=0.1cm]
\tikzstyle{conflict}=[ultra thick]

\node[clique] (gf) at (-1,0) {};
\node[clique] (gf2) at (1,0) {};
\node[ind] (fu) at (1,2) {};
\node[ind] (fv2) at (-1,2) {};
\node () at (-1,-0.5) {$g_{f,1}$};
\node () at (1,-0.5) {$g_{f,2}$};
\node () at (1,2.5) {$f_u$};
\node () at (-1,2.5) {$f_{v_2}$};

\draw (gf) -- (fu);
\draw (gf) -- (fv2);
\draw (gf2) -- (fu);
\draw (gf2) -- (fv2);

\node[clique] (gh) at (3,0) {};
\node[clique] (gh2) at (5,0) {};
\node[ind] (hu) at (5,2) {};
\node[ind] (hv3) at (3,2) {};
\node () at (3,-0.5) {$g_{h,1}$};
\node () at (5,-0.5) {$g_{h,2}$};
\node () at (5,2.5) {$h_u$};
\node () at (3,2.5) {$h_{v_3}$};

\draw (gh) -- (hu);
\draw (gh) -- (hv3);
\draw (gh2) -- (hu);
\draw (gh2) -- (hv3);

\node[clique] (ge) at (-5,0) {};
\node[clique] (ge2) at (-3,0) {};
\node[ind] (eu) at (-3,2) {};
\node[ind] (ev1) at (-5,2) {};
\node () at (-5,-0.5) {$g_{e,1}$};
\node () at (-3,-0.5) {$g_{e,2}$};
\node () at (-3,2.5) {$e_u$};
\node () at (-5,2.5) {$e_{v_1}$};

\draw (ge) -- (eu);
\draw (ge) -- (ev1);
\draw (ge2) -- (eu);
\draw (ge2) -- (ev1);

\draw[conflict] (ev1) -- (gf);
\draw[conflict] (ev1) -- (gh);

\draw[conflict] (fv2) -- (ge);
\draw[conflict] (fv2) -- (gh2);

\draw[conflict] (hv3) -- (ge2);
\draw[conflict] (hv3) -- (gf2);

\end{tikzpicture}
\end{tabular}

\end{center}
\caption{Construction when $u$ is an AND vertex (top) or an OR vertex (bottom). In both cases $v_1$ is a COPY vertex. 
The part of the construction corresponding to $\ell$ is not drawn: $\ell_{v_4}$ would be a common neighbor of $g_{e,1}, g_{e,2}$ and $e_u$ would be a common neighbor of $\ell_{e,1},\ell_{e,2}$. Edges connecting selector vertices 
to their corresponding gates are drawn thinner for readability. On the right, black (gate) vertices are connected in a clique.}
\label{fig:ANDOR}
\end{figure}

We are now ready to show that if we only consider ``consistent'' configurations
in $G_b$, then the new instance simulates the original NCL reconfiguration
problem.

\begin{lemma}\label{lem:step1}

There is a valid reconfiguration of the NCL instance given by $G,D,D'$ if and
only if there exists a valid reconfiguration under the TS rule  from $S$ to $T$
in $G_b$ such that no independent set of the reconfiguration sequence contains
both $e_u,e_v$ for any $e=(u,v)\in E$.

\end{lemma}

\begin{proof}

Since $G_b$ is a split graph, any independent set contains at most one vertex
from the clique made up of the gate vertices. We will call an independent set
that contains no gate vertices a ``main'' configuration. Furthermore, for main
configurations that also obey the restrictions of the lemma (i.e. do not
contain both $e_u,e_v$ for any $e\in E$), we observe that there is a natural
one-to-one correspondence with the set of orientations of $G$: an edge
$e=(u,v)$ is oriented towards $u$ if and only if $e_u$ is in the independent
set. (We implicitly use the fact that the number of tokens is $|E|$, therefore
for each pair $e_u,e_v$ exactly one vertex has a token in such a main
configuration).

Suppose now that we have two consecutive valid orientations $D_i, D_{i+1}$ in
the reconfiguration sequence of $G$ such that $D_i, D_{i+1}$ differ only on the
edge $e=(u,v)$, which $D_i$ orients towards $u$. We want to show that the sets
$I_i, I_{i+1}$ obtained using the correspondence above from $D_i, D_{i+1}$ can
be obtained from each other with a pair of sliding token moves. Indeed, the
sets $I_i, I_{i+1}$ are identical except that $\{e_u\}= I_i\setminus I_{i+1}$ and
$\{e_v\} = I_{i+1}\setminus I_i$. We would like to slide the token from $e_u$ to
$e_v$ using a gate vertex adjacent to both vertices. 

First, assume that $e\in R$, so there exists a single gate vertex $g_e$.
Furthermore, $u,v$ are both AND vertices. Since both $D_i, D_{i+1}$ are valid
configurations, in both configurations the blue edges incident on $u,v$ are
oriented towards these two vertices. As a result $g_e$ has no neighbor in
$I_i$. 

Second, suppose $e=(u,v)\in B$ and one of $u,v$ is a COPY vertex. If $e$ is
incident on an AND vertex, because both $D_i, D_{i+1}$ are valid and agree on
all edges except $e$ we have that both red edges incident on the AND vertex are
oriented towards it in both configurations. Similarly, the second blue edge
incident on the COPY endpoint of $e$ is oriented towards it in both
configurations. We therefore observe that neither $g_{e,1}$, nor $g_{e,2}$ has
a neighbor in $I_i$ except $e_u$, so we can safely slide $e_u\to g_{e,1}\to
e_v$.

Similarly, for the last case, suppose that $e=(u,v)\in B$ and one of the
endpoints of $e$ is an OR vertex, while the other is a COPY vertex. Again,
because $D_i, D_{i+1}$ are both valid and only disagree on $e$, at least one of
the blue edges incident on the OR vertex (other than $e$) is oriented towards
it in both configurations. As before, the second blue edge incident on the COPY
vertex is oriented towards it in both configurations.  Therefore, one of
$g_{e,1}, g_{e,2}$ has no neighbor in $I_i$ except $e_u$, so we can safely
slide the token from $e_u$ to $e_v$ with two moves.

To complete the proof, we need to show that if we have a valid token sliding
reconfiguration sequence, this gives a valid reorientation sequence for $G$.
The main observation now is that in a shortest token sliding solution that
obeys the properties of the lemma, a token that slides out of $e_u$ must
necessarily in the next move slide into $e_v$, where $e=(u,v)\in E$. To see
this, observe that because of the requirement that the set does not contain
both selector vertices of any edge, the tokens found on other selector vertices
dominate all gate vertices except those corresponding to $e$. Since we can
neither repeat configurations, nor add a second token to the clique made up of
gate vertices, the next move must slide the token to the other selector vertex. 

To see that the orientation sequence obtained through the natural translation of main
configurations is valid, consider two consecutive main configurations
$I_i,I_{i+1}$ in the token sliding solution, such that the corresponding
orientations are $D_i, D_{i+1}$, and $D_i$ is valid. We will show that
$D_{i+1}$ is also valid.  Suppose that $D_{i+1}$ differs from $D_i$ in the edge
$e=(u,v)$ which is oriented towards $u$ in $D_i$ (it is not hard to see that
$D_i, D_{i+1}$ cannot differ in more than one edge).  Thus, $I_i$ is
transformable in two moves to $I_{i+1}$ by sliding $e_u$ to a gate
corresponding to $e$ and then to $e_v$.  If $e$ is a red edge, this means that
in $D_i$ both blue edges incident on $u,v$ are directed towards $u,v$, so the
reorientation is valid. If $e$ is blue, we first assume that $u$ is a COPY
vertex.  Since a gate corresponding to $u$ is free, the other blue edge
incident on $u$ is oriented towards $u$ in $D_i$ and we have a valid move.
Finally, if $e$ is blue and $u$ is an OR vertex, we conclude that, since at
least one gate from $g_{e,1}, g_{e,2}$ is available in $I_i$, at least one of
the two other blue edges incident on $u$ is directed towards $u$ in $D_i$ and
we have a valid move.  \end{proof}

\subsubsection*{Second Step: Enforcing Consistency}

We will now construct a graph $G_f$ that will function in a way similar to the
graph we have already constructed but in a way that enforces consistency. Let
$G_b=(V_b,E_b)$ be the graph constructed in the first step of our reduction, and
let $E_g\subseteq E_b$ be the set of gate edges, that is, the set of edges that
connect the selector vertices for an edge $e$ to the corresponding gate(s).

Let $m:= |E|$ and $C := m+4$. We first take $C$ disjoint copies of $G_b=(V_b,E_b)$ and
for a vertex $v\in V_b$ we will use the notation $v^i$, where $1\le i\le C$ to
denote the vertex corresponding to $v$ in the $i$-th copy. Then, for every edge
$(u,v)\in E_b\setminus E_g$ (every non-gate edge) and for all
$i,j\in\{1,\ldots,C\}$ with $i \ne j$,  we add the edge $(u^i,v^j)$. This completes the
construction of $G_f$ and it is not hard to see that the graph is split, as the
$C$ copies of the clique of $G_b$ form a larger clique.  To complete our
instance let us explain how to translate an independent set of $G_b$ that
contains no vertices of the clique to an independent set of $G_f$: we do this
in the natural way by including in the new independent set all $C$ copies of
vertices of the original independent set.  Since both the initial and final
independent sets in our first construction use no vertices in the clique, we
have in this way two independent sets of size $mC$ in the new graph, and thus a
valid Token Sliding instance.  Let $S,T$ be the two independent sets of $G_b$
we are asked to transform and $S_f,T_f$ the corresponding independent sets of
$G_f$.

We first show that if we have a solution for reconfiguration in $G_b$ then we
have a solution for reconfiguring the sets in the new graph.

\begin{lemma}\label{lem:step2a}

Let $I_1,I_2$ be two independent sets of $G_b$ of size $m$ that use no vertices
of the clique, respect the conditions of Lemma \ref{lem:step1}, and can be
transformed to one another by two sliding moves. Then the independent sets
$I_1',I_2'$ which are obtained in $G_f$ by including all copies of vertices of
$I_1,I_2$ respectively can be transformed into one another by a sequence of
$2C$ TS moves.

\end{lemma}

\begin{proof}

Each of $I_1,I_2$ uses exactly one of the vertices $e_u,e_v$, for each edge
$e=(u,v)\in E$, because of their size, the fact that they contain no vertex of
the clique, and the fact that neither contains both $e_u,e_v$ for any edge
$e=(u,v)\in E$ (this is the condition of Lemma \ref{lem:step1}). If $I_1$ can
be transformed into $I_2$ with two sliding moves, the first move takes a token
from an independent set vertex, say $e_u$ and moves it to the clique and the
second moves the same token to $e_v$.  Since $I_1$ contains a token on each
pair of selector vertices, the only vertex of the clique on which the token can
be moved is a gate vertex corresponding to $e$, say $g_e$ (if $e$ is red) or
$g_{e,1}$ (if $e$ is blue).  We now observe that if $g_e$ (or similarly
$g_{e,1}$) is available in $I_1$ (that is, it has no neighbors in $I_1$ besides
$e_u$), then the same is true for $g_e^i$ for all $i\in\{1,\ldots,C\}$ in
$I_1'$. To see this, note that the neighbors of $g_e^i$ are, $e_u^i, e_v^i$,
and, for each $v\in N(g_e)$ all the vertices $v^j$ for $j\in\{1,\ldots,C\}$.
Since none of the neighbors of $g_e$ is in $I_1$, $g_e^i$ is available. We
therefore slide, one by one, a token from $e_u^i$ to $g_e^i$ and then to
$e_v^i$, for all $i\in\{1,\ldots,C\}$.  \end{proof}

Now, for the more involved direction of the reduction we first observe that it
is impossible for a reconfiguration to arrive at a situation where the solution
is highly irregular, in the sense that, for an edge $e=(u,v)$ we have multiple
tokens on copies of both $e_u$ and $e_v$. 

\begin{lemma}\label{lem:step2b}

Let $S_f$ be the initial independent set constructed in our instance and $S'$
be an independent set which for some $e=(u,v)\in E$ and for some $i,j\in
\{1,\ldots,C\}$ with $i\neq j$ has $e_u^i,e_v^i, e_u^j, e_v^j\in S'$. Then $S'$
is not reachable with TS moves from $S_f$.

\end{lemma}

\begin{proof}

Let $S'$ be an independent set that satisfies the conditions of the lemma but
is reachable from $S_f$ with the minimum number of token sliding moves.
Consider a sequence that transforms $S_f$ to $S'$, and let $S''$ be the
independent set immediately before $S'$ in this sequence. $S''$ contains
exactly three of the vertices $e_u^i,e_v^i, e_u^j, e_v^j$. Without loss of
generality say $e_v^j\not\in S''$. Therefore, the move that transforms $S''$ to
$S'$ slides a token into $e_v^j$ from one of the neighbors of this vertex. We
now observe that $N(e_v^j)$ contains $C$ copies of each neighbor of $e_v$ in
$G_b$, plus the gate vertices corresponding to $e$ in the $j$-th copy of $G_b$.
However, the $C$ copies of the neighbors of $e_v$ are also neighbors of
$e_v^i$, hence a token cannot slide through these vertices. Furthermore, the
gate vertices of $e$ are also neighbors of $e_u^j$. We therefore have a
contradiction.  \end{proof}

We now use Lemma \ref{lem:step2b} to show that for each original edge, the graph
$G_f$ contains some non-trivial number of tokens on the selector vertices of
that edge.

\begin{lemma}\label{lem:step2c} 

Let $S_f$ be the initial independent set constructed in our instance and $S'$
be an independent set which for some $e=(u,v)\in E$ has $|S'\cap (\{e_u^i\ |\
1\le i \le C\} \cup \{e_v^i\ |\ 1\le i \le C\})| < 4  $.  Then $S'$ is
unreachable from $S_f$.

\end{lemma}

\begin{proof}

Suppose $S'$ is reachable. Then by Lemma \ref{lem:step2b}, for each edge
$e=(u,v)\in E$ we have $|S'\cap (\{e_u^i\ |\ 1\le i \le C\} \cup \{e_v^i\ |\
1\le i \le C\})| \le C+1  $, because otherwise there would exist (by pigeonhole
principle) $e_u^i,e_v^i,e_u^j,e_v^j\in S'$. We now use a simple counting
argument. The total number of tokens is $mC$, while for any edge $f\in E$ we
have $\sum_{e\in E\setminus \{f\}} |S'\cap (\{e_u^i\ |\ 1\le i\le C\} \cup
\{e_v^i\ |\ 1\le i\le C\})| \le (m-1)(C+1)$. However, $(m-1)(C+1) = mC + m - C
- 1 = mC - 5$, where we use the fact that $C=m+4$.  As a result $|S'\cap
(\{e_u^i\ |\ 1\le i \le C\} \cup \{e_v^i\ |\ 1\le i \le C\})| \ge 4$ for any
edge $e\in E$, as the independent set $S'$ uses at most one vertex from the
clique.  \end{proof}

We are now ready to establish the final lemma that gives a mapping from a
sliding token reconfiguration in $G_f$ to one in $G_b$.

\begin{lemma}\label{lem:step2d} If there exists a reconfiguration from $S_f$ to
$T_f$ in $G_f$ under the TS rule then there exists a reconfiguration from $S$
to $T$ in $G_b$ under the TS rule which for each edge $e=(u,v)\in E$ contains
at most one of the vertices $e_u,e_v$ in every independent set in the sequence.
\end{lemma}

\begin{proof}

Take a configuration $I$ of $G_f$, that is an independent set in the supposed
sequence from $S_f$ to $T_f$. We map this independent set to an independent set
$I'$ of $G_b$ as follows: for each edge $e=(u,v)\in E$, we set $e_u\in I'$ if
and only if $|I\cap \{e_u^i\ |\ 1\le i\le C\}| \ge |I\cap \{e_v^i\ |\ 1\le i
\le C\}|$.  Informally, this means that we take the majority setting from
$G_f$.  We note that this always gives an independent set $I'$ that contains
exactly one vertex from $\{e_u, e_v\}$ for each $e=(u,v)\in E$.

Our main argument now is to show that if $I_1,I_2$ are two consecutive
independent sets of the solution for $G_f$, then the sets $I_1',I_2'$ which are
obtained in the way described above in $G_b$ are either identical or can be
obtained from one another with two sliding moves. If $I_1',I_2'$ are not
identical, they may differ in at most two vertices corresponding to an edge
$e=(u,v)\in E$, say $\{e_u\}= I_1'\setminus I_2'$ and $\{e_v\}= I_2'\setminus
I_1'$. This is not hard to see, since $I_2$ is obtained from $I_1$ with one
sliding move, and this move can only affect the majority opinion for at most
one edge.

Now we would like to argue that it is possible to slide $e_u$ to a gate vertex
associated to $e$ and then to $e_v$ in $G_b$. Consider the transition from
$I_1$ to $I_2$. This move either slides a token from some $e_u^i$ to the
clique, or slides a token from the clique to some $e_v^j$ (because the majority
opinion changed from $e_u$ to $e_v$). Because of Lemma \ref{lem:step2c}, both
$I_1$ and $I_2$ contain at least four vertices in some copies of $e_u,e_v$.
Hence, since at least half of these vertices are in copies of $e_u$ in $I_1$,
there exists some $e_u^i\in I_1\cap I_2$. Similarly, there exists some
$e_v^j\in I_1\cap I_2$. Consider now a gate vertex $g$ in the clique of $G_b$
such that $g$ is not associated with $e$. If $g$ has an edge to $\{e_u,e_v\}$
in $G_b$, then all copies of $g$ in $G_f$ have an edge to $I_1\cap I_2$,
therefore cannot belong in either set. As a result, the clique vertex that is
used in the transition from $I_1$ to $I_2$ is a copy of a gate vertex
associated with $e$ (either $g_e$, or one of $g_{e,1},g_{e,2}$, depending on
the color of $e$).  This gate vertex copy therefore has no neighbor in $I_1\cap
I_2$. From this we conclude that the same gate vertex in $G_b$ also has no
neighbor in $I_1'\cap I_2'$, as the majority opinion only changed for $e$. It
is therefore legal to slide from $e_u$ to this gate vertex and then to $e_v$.
\end{proof}

\begin{theorem}
\label{thm:split-TS}

Sliding Token Reconfiguration is PSPACE-complete for split graphs.

\end{theorem}

\begin{proof}

We begin with an instance of the PSPACE-complete NCL reconfiguration problem,
as given in Lemma \ref{lem:mycsl}. We construct the instance $G_f, S_f, T_f$ of
Sliding Token Reconfiguration on split graphs as described (it's clear that
this can be done in polynomial time). If the NCL reconfiguration instance is a
YES instance, then by Lemma \ref{lem:step1} there exists a sliding token
reconfiguration of $G_b$, and by repeated applications of Lemma
\ref{lem:step2a} to independent sets that do not contain clique vertices in the
reconfiguration of $G_b$ there exists a sliding token reconfiguration of $G_f$.
If on the other hand there exists a sliding token reconfiguration on $G_f$,
then by Lemma \ref{lem:step2d} there exists a reconfiguration that satisfies
the condition of Lemma \ref{lem:step1} on $G_b$, hence the original NCL
instance is a YES instance.  \end{proof}

%% file: chordal.tex
\section{PSPACE-completeness for Chordal Graphs for $c\ge2$}

In this section, we build upon the PSPACE-completeness result from
Section~\ref{sec:PSPACE-split} to show that $c$-\textsc{Colorable Set
Reconfiguration} is PSPACE-complete, for every $c\geq 2$, when the input graph
is restricted to be chordal. 
%This contrasts sharply with the results obtained in
%Section~\ref{sec:strongly-chordal}.

\begin{theorem}\label{thm:TS-split2}
For every $c\geq 2$, the $c$-\textsc{Colorable Set Reconfiguration} problem under the TS rule is PSPACE-complete,  even when the input graph is restricted to be chordal.
\end{theorem}
\begin{proof}
%[Poof of Theorem \ref{thm:TS-split2}]

We provide a reduction from \textsc{Independent Set Reconfiguration} where the
input graph $G$ is restricted to be a split graph, which we proved to be
PSPACE-complete in Theorem \ref{thm:split-TS}. Let $G=(V,E)$ be an input split
graph for \textsc{Independent Set Reconfiguration}.  We construct a chordal
graph $G'$ as follows, starting from a graph isomorphic to $G$ and two
non-empty independents set $S,T$ of the same size. For every edge $uv\in E(G)$,
we add $|V(G)|$ sets of $c-1$ new vertices $W_{uv}^1,\ldots,W_{uv}^{|V(G)|}$,
such that $W_{uv}^i$ induces a clique for every $1\leq i\leq |V(G)|$, and every
vertex of $W_{uv}^i$ is made adjacent to both $u$ and $v$, for every $1\leq
i\leq |V(G)|$. In addition, we create a new set $S'=S\cup \bigcup_{uv\in E(G),
1\leq i\leq |V(G)|}W_{uv}^i$ and a set $T'=T\cup \bigcup_{uv\in E(G), 1\leq
i\leq |V(G)|}W_{uv}^i$.  In other words, we append $|V(G)|$ disjoint cliques of
size $c-1$ to every edge of $G$, and add all those newly created vertices to
$S$ and to $T$.  The chordality of $G'$ follows from the fact that the new
vertices of the sets $W_{uv}^i$ are all simplicial in $G'$, hence $G'$ is
chordal if and only if $G$ is chordal as well (and $G$ is split).

We now claim the following: given in independent set $T$ of $G$, the instance
$(G,S,T)$ of \textsc{Independent Set Reconfiguration} is a YES-instance if and
only if the instance $(G',S',T')$ of $c$-\textsc{Colorable Set Reconfiguration}
is a YES-instance as well.  Observe that, by the construction,
$S'$ and $T'$ are $c$-colorable because the maximum
clique in $G'[S']$ contains at most one vertex of $S$ and at most the $c-1$
vertices of a clique $W_{uv}^i$.

The forward direction of the previous claim follows easily: performing the
same moves as those of a reconfiguration sequence from $S$ to $T$ in $G'$,
starting from $S'$, yields a reconfiguration sequence where every step
preserves $c$-colorability, and produces the desired set $T'$.

For the backwards direction, we claim that, for any $c$-colorable set $R'$
reachable from $S'$, it holds that the vertices of $R' \cap V(G)$ are pairwise
non-adjacent. In other words, the tokens placed on original vertices of $G$
form an independent set. 
Indeed, observe that the number of vertices of $G'$ that do not belong to $R'$ satisfies 
$|V(G') \setminus R'| = |V(G) \setminus S| < |V(G)|$. This immediately implies that for
any set $R'$ and edge $uv\in E(G)$, we have $|R' \cap \bigcup_{1\leq i\leq |V(G)|}W_{uv}^i| \geq (c-2)|V(G)|+1$, 
and therefore $G[R' \cap\bigcup_{1\leq i \leq |V(G)|}W_{uv}^i]$ contains a clique of size $c-1$ as an induced subgraph,
i.e., one of the sets $W_{uv}^i$ is completely contained in $R'$. This implies
that, for every edge $uv$ of $G$, we have $|R'\cap \{u,v\}| \leq 1$, i.e., the
vertices of $R' \cap V(G)$ are pairwise non-adjacent, as desired. 
\end{proof}

%% file: xpalgo.tex
%!TEX root = reconfiguration.tex

\section{XP-time Algorithm on Split Graphs for fixed $c\geq 2$}\label{sec:xp}

In this section we present an $n^{O(c)}$ algorithm for $c$-\textsc{Colorable
Reconfiguration} under the TS rule, on split graphs, for $c>1$.  Recall that a
split graph $G=(V,E)$ is a graph whose vertex set $V$ is partitioned into a
clique $K$ and an independent set $I$.  An input instance consists of a split
graph $G$, and two $c$-colorable sets $S,T\subseteq V$. 

Before proceeding, let us give some high-level ideas as well as some intuition
why this problem, which is PSPACE-complete for $c=1$ (Theorem
\ref{thm:split-TS}), admits such an algorithm for larger $c$. Our algorithm
consists of two parts: a rigid and a non-rigid reconfiguration part. In the
rigid reconfiguration part the algorithm decides if two sets are reachable by
using moves that never slide tokens into or out of $I$. Because of this
restriction and the fact that the sets are $c$-colorable, the total number of
possible configurations is $n^{O(c)}$, so this part can be solved with
exhaustive search (this is similar to the algorithm of \cite{ItoO18} for
TJ/TAR).  In the non-rigid part we assume we are given two sets $S,T$ which, in
addition to being $c$-colorable, have $|S\cap K|,|T\cap K|\le c-1$. The main
insight is now that \emph{any} two such sets are reachable via TS moves (Lemma
\ref{lem:smallconnected} below). Informally, the algorithm guesses a partition
of the optimal reconfiguration into a rigid prefix, a rigid suffix, and a
non-rigid middle, and uses the two parts to calculate each independently.

The intuitive reason that our algorithm cannot work for $c=1$ is the non-rigid
part.  The crucial Lemma \ref{lem:smallconnected} on which this part is based
fails for $c=1$: for instance, if $G$ is a star with three leaves and $S,T$ are
two distinct sets each containing two leaves, then $S,T$ satisfy all the
conditions for $c=1$, but are not reachable from each other with TS moves.
Such counterexamples do not, however, exist for higher $c$, because for sets
that satisfy the conditions of Lemma \ref{lem:smallconnected} we know we can
always freely move tokens around inside the clique (and without loss of
generality, such tokens exist). Note also, that this difficulty is specific to
the TS rule: the algorithm of \cite{ItoO18} implicitly uses the fact that any
two sets with $c-1$ tokens in the clique are always reachable, as this is an
almost trivial fact if one is allowed to use TJ moves. Thus, Lemma
\ref{lem:smallconnected} is the main new ingredient that makes our algorithm
work.

Let us now proceed with a detailed description of the algorithm. First, let us
fix some notation. For a vertex set $R\subseteq V$, we write the subsets $R\cap K$
and $R\cap I$ as $R_K$ and $R_I$ respectively.
%
%A few assumptions can be
%made on any input $(G,S,T)$ of $c$-\textsc{Colorable Token Sliding
%Reconfiguration}, which we will maintain throughout this section.
%
Throughout this section, we assume that 
input graph $G = (K \cup I, E)$ is connected (and thus each vertex in $I$ has a neighbor in $K$);
otherwise we can consider instances induced by each component separately.  

%\begin{enumerate} 
%
%\item $G$ is connected; otherwise, we consider instances induced by each
%component separately.  
%
%%\item $|S|=|T|$; otherwise, it is a NO-instance.  
%
%%\item For every $v\in I$, there exists $u\in K$ such that $u$ and $v$ are
%%non-adjacent; otherwise, $v$ can be re-classified as a vertex of $K$.  
%
%\item Every $v\in I$ has a neighbor in $K$; if not, an input $(G,S,T)$ can be
%simplified. That is, if precisely one of $S$ and $T$ contains $v$, then the
%input is a trivial NO-instance.  Otherwise, no TS move in a reconfiguration
%sequence from $S$ to $T$ involves $v$. Therefore, $(G-v,S\setminus v,
%T\setminus v)$ is an equivalent instance. 
%
%\end{enumerate}

\begin{lemma}\label{lem:smallconnected}
Let $G$ be a split graph, $c\ge2$, and $S,T\subseteq V$ be two $c$-colorable sets such that $|S_K|,|T_K|\leq c-1$. 
Then $T$ is $c$-reachable from $S$. Furthermore, a reconfiguration sequence from $S$ to $T$ 
can be produced in polynomial time.
\end{lemma}
\begin{proof}
%Without loss of generality, we assume $|S_K|\leq |T_K|$. 
We first observe that if $S_I=T_I$, then there is an easy optimal $c$-transformation. By making one TS move from $u\in S_K\setminus T_K$ to $v\in T_K\setminus S_K$, 
one can $c$-transform $S$ to $T$ with $|S\setminus T|$ sliding moves (thus yielding an optimal reconfiguration sequence). 
It is clear that all the sets resulting from these TS moves are $c$-colorable because each of them has at most $c-1$ vertices in $K$. 

Therefore, it suffices to show that there is always a $c$-transformation of 
$T$ which decrease $|S_I \setminus T_I|$ as long as $S\neq T$.  Note that we
can assume that there exists $v\in S_I\setminus T_I$ (otherwise we exchange the
roles of $S$ and $T$).  In the case when $T_K=\emptyset$, one can transform $T$
to $T'$ with TS moves from a vertex of $T_I\setminus S_I$ to $v$. Trivially
this is a $c$-transformation, and it holds that $|T'_K|=\emptyset$. (Note that
this argument would not be valid if $c=1$).  If $T_K\neq \emptyset$, then one
can make at most two TS moves from a vertex of $T_K$ to $v$. Because $T$ has at
most $c-1$ vertices and these TS moves maintain at most $c-1$ vertices in $K$,
$c$-colorability of $T$ is preserved.  Moreover, the new set has at most $c-1$
vertices in $K$ while its intersection with $S$ in $I$ is strictly larger.
This completes the proof of the first statement. The proof is constructive and
easily translates to a polynomial-time algorithm.  \end{proof}

Let us now introduce a notion that will be useful in our algorithm. For two
$c$-colorable sets $S,T$ with $S_I=T_I$ we say that $S$ has a \emph{rigid}
$c$-transformation to $T$ if there exists a valid $c$-transformation from $S$
to $T$ with TS moves which also has the property that every $c$-colorable set
$R$ of the transformation has $R_I=S_I$.

\begin{lemma}\label{lem:rigid} Given a split graph $G=(V,E)$, with $V=K\cup I$,
and two $c$-colorable sets $S,T\subseteq V$ with $S_I=T_I$, there is an
algorithm that decides if there exists a rigid $c$-transformation of $S$ to $T$
in time $n^{O(c)}$.  \end{lemma}

\begin{proof}

The main observation is that since all intermediate sets must have $R_I=S_I$,
we are only allowed to slide tokens inside $K$. However, $S_K$ contains at most
$c$ vertices (as it is $c$-colorable), therefore, there are at most $n^c$
potentially reachable sets: one for each collection of $|S_K|$ vertices of the
clique.

We now construct a secondary graph with a node for each subset of $V$ that
contains $|S_K|$ vertices of $K$ and the vertices of $S_I$, and connect two
such nodes if their corresponding sets are reachable with a single TS move in
$G$. In this graph we check if there is a path from the node that represents
$S$ to the one that represents $T$ and if yes output the sets corresponding to
the nodes of the path as our rigid reconfiguration sequence.  \end{proof}

\begin{theorem}\label{thm:xpalg}

There is an algorithm that decides $c$-\textsc{Colorable Reconfiguration} on
split graphs under the TS rule in time $n^{O(c)}$, for $c\ge2$.

\end{theorem}

\begin{proof}

We distinguish the following cases: (i) $|S_K|,|T_K|\le c-1$, (ii) $|S_K|=c$
and $|T_K|=c-1$, (iii) $|S_K|=|T_K|=c$. This covers all cases since $S,T$ are
$c$-colorable and we can assume without loss of generality that $|S_K|\ge
|T_K|$.

For case (i) we invoke Lemma \ref{lem:smallconnected}. The answer is always
Yes, and the algorithm of the lemma produces a feasible reconfiguration
sequence.

For case (ii), suppose there exists a reconfiguration sequence from $S$ to
$T$, call it $T_0=S, T_1,\ldots, T_{\ell}=T$. Let $i$ be the smallest index
such that $|T_i\cap K|\le c-1$. Clearly such an index exists, since $|T_K|\le
c-1$. We now guess the configuration $T_{i-1}$ and the configuration $T_i$
(that is, we branch into all possibilities). Observe that there are at most
$n^c$ choices for $T_{i-1}$ as we have $T_{i-1}\cap I = S_I$ and $|T_{i-1}\cap
K|=c$.  Furthermore, once we have selected a $T_{i-1}$, there are $n^{O(1)}$
possibilities for $T_i$, as $T_i$ is reachable from $T_{i-1}$ with one TS move.

We observe that if we guessed correctly, then there exists a rigid
$c$-transformation from $S$ to $T_{i-1}$ (by the minimality of $i$ and the fact
that $|S_K|=c$); we use the algorithm of Lemma \ref{lem:rigid} to check this.
Furthermore, the configuration $T_i$ is always transformable to $T$ by Lemma
\ref{lem:smallconnected}. Therefore, if the algorithm of Lemma \ref{lem:rigid}
returns a solution, then we have a $c$-transformation from $S$ to $T$.
Conversely, if a $c$-transformation from $S$ to $T$ exists, since we tried all
possibilities for $T_{i-1}$, one of the branches will find it.

Finally, for case (iii), if $S_I=T_I$ we first use Lemma \ref{lem:rigid} to
check if there is a rigid $c$-transformation from $S$ to $T$. If one is found,
we are done. If not, or if $S_I\neq T_I$ we observe that, similarly to case
(ii), in any feasible transformation  $T_0=S, T_1,\ldots, T_{\ell}=T$, there
exists an $i$ such that $|T_i\cap K|\le c-1$ (otherwise the transformation
would be rigid). Pick the minimum such $i$. We now guess the configurations
$T_{i-1},T_i$ (as before, there are $n^{c+O(1)}$ possibilities) and use Lemma
\ref{lem:rigid} to verify that $T_{i-1}$ is reachable from $S$. If $T_{i-1}$ is
reachable from $S$, we need to verify that $T$ is reachable from $T_i$.
However, we observe that this reduces to case (ii), because $|T_i\cap K|\le
c-1$, so we proceed as above.  If the algorithm returns a valid sequence we
accept, while we know that if a valid sequence exists, then there exists a
correct guess for $T_{i-1},T_i$ that we consider.  \end{proof}

%% file: whard.tex
\section{W-hardness for Split Graphs}

In this section we show that $c$-\textsc{Colorable Reconfiguration} on split
graphs  is W[2]-hard parameterized by $c$ and the length $\ell$ of
the reconfiguration sequence under all three reconfiguration rules (TAR, TJ,
and TS). In this sense, this section complements Section
\ref{sec:xp} by showing that the $n^{O(c)}$ algorithm that we presented
for $c$-\textsc{Colorable Reconfiguration} on split graphs cannot be
significantly improved under standard assumptions.

We will rely on known results on the hardness of \textsc{Dominating Set
Reconfiguration}. We recall that in this problem we are given a graph
$G=(V,E)$, two dominating sets $S,T\subseteq V$ of size at most $k$ and are
asked if we can transform $S$ into $T$ by a series of TAR operations while
keeping the size of the current set at most $k$ at all times.  More formally,
we are asked if there exists a sequence $T_0=S, T_1,\ldots, T_{\ell}=T$ such
that for each $i\in \{0,\ldots,\ell-1\}$, $|T_i|\le k$, $T_i$ is a dominating
set of $G$, and $|(T_i\setminus T_{i+1}) \cup (T_{i+1}\setminus T_i)|=1$.

\begin{theorem}[\cite{MouawadN0SS17}]\label{thm:dshard1} \textsc{Dominating Set
Reconfiguration} is W[2]-hard parameterized by the maximum size of the allowed
dominating sets $k$ and the length $\ell$ of the reconfiguration sequence under
the TAR rule.

\end{theorem}

Before proceeding, let us make two remarks on Theorem \ref{thm:dshard1}: first,
because the reduction of \cite{MouawadN0SS17} is linear in the
parameters, it is not hard to see that it also implies a tight ETH-based lower
bound based on known results for \textsc{Dominating Set}; second, using an
argument similar to that of Theorem 1 of \cite{KaminskiMM12}, the same hardness
can be obtained for the TJ rule.

\begin{corollary}\label{cor:dshard2} \textsc{Dominating Set Reconfiguration} is
W[2]-hard parameterized by the maximum size of the allowed dominating sets $k$
and the length $\ell$ of the reconfiguration sequence under the TAR, or TJ
rule.  Furthermore, the problem does not admit an algorithm running in
$n^{o(c+\ell)}$ under the ETH for any of the two rules. \end{corollary}

\begin{proof}
%[Poof of Corollary \ref{cor:dshard2}]
To obtain hardness under the TJ rule
we use an argument similar to that of Theorem 1 of \cite{KaminskiMM12}. Suppose
we are given an instance of $k$-\textsc{Dominating Set Reconfiguration}
$G=(V,E)$ and $S,T\subseteq V$ where $k$ is the maximum size of any dominating
set allowed and we use the TAR rule, that is, an instance produced by the
reduction establishing Theorem \ref{thm:dshard1}. We recall that in the
instances produced for this reduction we have $k=\Theta(\ell)$ and that $S$ can
be transformed into $T$ with $\ell$ TAR moves if and only if $S$ can be
transformed into $T$ with some number of TAR moves (in other words, if $\ell$
moves are not sufficient, then $S$ and $T$ are in fact unreachable). This
observation will be useful because it means that in the reduction that follows
we do not have to preserve $\ell$ exactly but only guarantee that it increases
by at most a constant factor.

We can assume without loss of generality that $|S|=|T|=k-1$: if $|S|<k-1$ we
can add to $S$ arbitrary vertices to make its size $k-1$, while if $|S|=k$ then
$S$ cannot be a minimal dominating set (otherwise it would be impossible to
transform it to any other set and we would have an obvious NO instance) so
there is a vertex that we can remove from $S$ without affecting the answer. In
both cases we appropriately increase $\ell$ by the number of modifications we
made to $S,T$ to preserve reachability. We want to show that the instance is
now equivalent under the TJ rule.  In particular, there exists a TAR
reconfiguration with $2\ell$ moves if there exists a TJ reconfiguration with
$\ell$ moves.

First, if there exists a TJ reconfiguration from $S$ to $T$ then there exists a
TAR reconfiguration from $S$ to $T$: for each move that exchanges $u\in S$ with
$v\not\in S$ we first add $v$ to $S$ and then remove $u$.

For the converse direction, suppose that there is a TAR reconfiguration of $S$
to $T$. If moves alternate in this reconfiguration, that is, if all
intermediate sets have size between $k-2$ and $k$, then it is not hard to see
how to perform the same reconfiguration with TJ moves.  Suppose then that the
reconfiguration performs two consecutive vertex removal moves, so we have the
dominating sets $T_i, T_{i+1}, T_{i+2}$ appearing consecutively in the
reconfiguration sequence, with $|T_i|=|T_{i+1}|+1=|T_{i+2}|+2$.  Let $j$ be the
smallest index with $j>i+2$ such that $|T_j|>|T_{j-1}|$ (i.e.  $j$ signifies
the first time we added a vertex after the $i$-th move). Let $T_i\setminus
T_{i+1}=\{u\}$ and $T_j\setminus T_{j-1}=\{v\}$. Then, if $u=v$ we can add $u$
to all sets $T_{i+1},\ldots,T_{j-1}$ and obtain a shorter reconfiguration
sequence (since now $T_i=T_{i+1}$ and $T_j=T_{j-1}$). Similarly, if $u\neq v$
and $v\in T_{i+1}$ we add $v$ to all sets $T_{i+2},\ldots,T_{j-1}$ to which it
doesn't appear and we have a shorter reconfiguration sequence. Finally, if
$u\neq v$ and $v\not\in T_{i+1}$, we insert after $T_{i+1}$ the set
$T_{i+1}\cup\{v\}$ and then add $v$ to all sets $T_{i+2},\ldots,T_{j-1}$.  We
now have $T_{j-1}=T_j$, so we have a valid TAR reconfiguration of the same
length but with one less pair of consecutive vertex removals. Repeating this
argument produces a TAR reconfiguration which can be performed with TJ moves.

For the ETH-based lower bound it suffices to recall that, under the ETH
$t$-\textsc{Dominating Set} does not admit an $n^{o(t)}$ algorithm
\cite{CyganFKLMPPS15}, and that the reduction establishing Theorem
\ref{thm:dshard1} in \cite{MouawadN0SS17} is a reduction from
$t$-\textsc{Dominating Set} that sets $k,\ell=O(t)$.  \end{proof}

\begin{theorem}\label{thm:whard} The $c$-\textsc{Colorable Reconfiguration}
problem is W[2]-hard parameterized by $c$ and the reconfiguration length $\ell$
when restricted to split graphs under any of the three reconfiguration rules
(TAR, TJ, TS).  Furthermore, under the ETH, the same problem does not admit an
$n^{o(c+\ell)}$ algorithm.  \end{theorem}

\begin{proof}
%[Poof of Theorem \ref{thm:whard}]

We use a reduction from \textsc{Dominating Set Reconfiguration} similar to the
one used in \cite{ItoO18} to prove that our problem is PSPACE-complete if $c$
is part of the input.  Let $G=(V,E)$ be an input graph for \textsc{Dominating
Set Reconfiguration}.  We construct a split graph $G'$ as follows: we take two
copies of $V$, call them $V_1,V_2$; we turn $V_1$ into a clique; for each $u\in
V_1$ and $v\in V_2$ we add the edge $(u,v)$ if and only if $u\not\in N[v]$ in
$G$. In other words, we connect each vertex from $V_1$ with all the vertices of
$V_2$ which it \emph{does not} dominate in $G$.

We assume now that we have started with $k$-\textsc{Dominating Set
Reconfiguration} instance under the TJ rule, which is W[2]-hard according to
Corollary \ref{cor:dshard2} parameterized by $k+\ell$. We will first show
hardness of $c$-\textsc{Colorable Reconfiguration} for TJ and TS parameterized
by $c+\ell$.

We construct a one-to-one correspondence between size $k$ dominating sets of
$G$ and $k$-colorable sets of vertices of $G'$ of size $n+k$, where $n=|V|$:
for each such set $S\subseteq V$ we define its image $\phi(S)$ in $G'$ as
$\{u\in V_1\ |\ u\in S\}\cup V_2$. In other words, we select all the vertices
of $S$ from $V_1$ and all of $V_2$. It is not hard to see that $\phi(S)$ is
indeed $k$-colorable: if not, there exists a clique of size $k+1$ in $G'[S']$
(since split graphs are perfect), which must consist of the $k$ vertices of $S$
from $V_1$, plus a vertex $v$ from $V_2$. But $v$ must be dominated by a vertex
$u\in S$ in $G$, which means that $v$ and the copy of $u$ in $V_1$ are not
connected.

Let us also observe that for every $k$-colorable set $S'$ of size $n+k$ in $G'$
we have that $S'=\phi(S)$ for some dominating set $S$ of size $k$ in $G$. To
see this, observe that $S'$ must contain exactly $k$ vertices of $V_1$ (since
it is $k$-colorable, $V_1$ is a clique, and $|V_2|=n$). These vertices must be
a dominating set of $G$ as otherwise there would exist a vertex $v$ that is not
in any of their closed neighborhoods, and the copy of $v$ in $V_2$ together
with $S'\cap V_1$ would form a clique of size $k+1$, contradicting the
$k$-colorability of $S'$.

Given the above correspondence it is not hard to complete the reduction: if we
are given two dominating sets $S,T\subseteq V$ with the initial instance we set
$\phi(S),\phi(T)$ as the two $k$-colorable graphs of the new instance. We
observe that any valid TJ move that transforms a dominating set $T_i$ to a
dominating set $T_{i+1}$ in $G$, corresponds to a TJ move that transforms
$\phi(T_i)$ to $\phi(T_{i+1})$ in $G'$. Crucially, such a move is also a TS
move, as the symmetric difference of $T_{i}$ and $T_{i+1}$ is contained in the
clique. Hence, there is also a one-to-one correspondence between TJ
$k$-dominating set reconfigurations in $G$ and TS $k$-colorable subgraph (of
size $n+k$) reconfiguration in $G'$. We therefore set the length of the desired
reconfiguration sequence in $G'$ to $\ell$.

Finally, to obtain hardness of the new instance under the TAR rule we set the
lower bound on the size of any intermediate set to $n+k-1$. Since
$|\phi(S)|=|\phi(T)|= n+k$ this means that any TJ $c$-colorable reconfiguration
can also be performed with at most $2\ell$ TAR moves. For the converse
direction we observe that in any TAR reconfiguration we never have a set of
size $n+k+1$ or more, since such a set would necessarily induce a graph that
needs $k+1$ colors. Hence, such a reconfiguration must consist of alternating
vertex removal and addition moves, which can be performed with $\ell$ TJ moves.

The ETH-based lower bounds follow from Corollary \ref{cor:dshard2} and the fact
that the reduction we performed is at most linear in all parameters.
\end{proof}